\documentclass[pdflatex,sn-basic, iicol]{sn-jnl}% Math and Physical Sciences Numbered Reference Style
\usepackage{graphicx}%
\usepackage{multirow}%
\usepackage{amsmath,amssymb,amsfonts}%
\usepackage{amsthm}%
\usepackage{mathrsfs}%
\usepackage[title]{appendix}%
\usepackage{xcolor}%
\usepackage{textcomp}%
\usepackage{manyfoot}%
\usepackage{booktabs}%
\usepackage{algorithm}%
\usepackage{algorithmicx}%
\usepackage{algpseudocode}%
\usepackage{listings}%

\usepackage{bm}
\theoremstyle{thmstyleone}%
\newtheorem{theorem}{Theorem}%  meant for continuous numbers
\newtheorem{proposition}[theorem]{Proposition}% 

\theoremstyle{thmstyletwo}%
\newtheorem{remark}{Remark}%

\theoremstyle{thmstylethree}%

\begin{document}

\title[Parameter Estimation of the Network of FitzHugh-Nagumo Neurons Based on the Speed-Gradient and Filtering]{Parameter Estimation of the Network of FitzHugh-Nagumo Neurons Based on the Speed-Gradient and Filtering}


\author*[1,2]{\fnm{Aleksandra} \sur{Rybalko}}\email{alexandrarybalko21@gmail.com}

\author[1,2]{\fnm{Alexander} \sur{Fradkov}}\nomail

\affil[1]{\orgdiv{Institute for Problems in Mechanical Engineering}, \orgname{RAS}, \orgaddress{\street{Bolshoy prospekt V.O., 61}, \city{Saint Petersburg}, \postcode{199178}, \country{Russia}}}

\affil[2]{\orgdiv{Theoretical Cybernetics Department}, \orgname{SPBU}, \orgaddress{\street{7-9 Universitetskaya Embankment}, \city{Saint Petersburg}, \postcode{199034}, \country{Russia}}}

\abstract{The paper addresses the problem of parameter estimation (or identification) in dynamical networks composed of an arbitrary number of FitzHugh–Nagumo neuron models with diffusive couplings between each other. It is assumed that only the membrane potential of each model is measured, while the other state variable and all derivatives remain unmeasured. Additionally, potential measurement errors in the membrane potential due to sensor imprecision are considered. To solve this problem, firstly, the original FitzHugh–Nagumo network is transformed into a linear regression model, where the regressors are obtained by applying a filter-differentiator to specific combinations of the measured variables. Secondly, the speed-gradient method is applied to this linear model, leading to the design of an identification algorithm for the FitzHugh–Nagumo neural network. Sufficient conditions for the asymptotic convergence of the parameter estimates to their true values are derived for the proposed algorithm. Parameter estimation for a network of five interconnected neurons is demonstrated through computer simulation. The results confirm that the sufficient conditions are satisfied in the numerical experiments conducted. Furthermore, the algorithm's capabilities for adjusting the identification accuracy and time are investigated. The proposed approach has potential applications in nervous system modeling, particularly in the context of human brain modeling. For instance, EEG signals could serve as the measured variables of the network, enabling the integration of mathematical neural models with empirical data collected by neurophysiologists.}


\maketitle

\section{Introduction}\label{sec1}

One of the main fields in modern neuroscience is mathematical and computer modeling of neural processes taking place in the nerve system, in particular, in the human brain. However, the widespread application of such modeling to real experimental data can be hindered by the complexity of certain detailed models and the presence of undefined model parameters. The former issue can be addressed by employing simplified models under pragmatic assumptions, such as the FitzHugh–Nagumo (FHN) model \citep{vittadello}. The latter issue arises because these parameters depend not only on the neural phenomenon under consideration but also on its realization in a specific organism (or a group of organisms) or over time. Consequently, the parameters cannot be fixed, necessitating the development of mathematical model parameter estimation (or identification) algorithms.

A major advantage of such mathematically grounded algorithms is their ability to guarantee accuracy under certain defined conditions. These guarantees, along with the corresponding conditions, can only be derived by mathematics, highlighting the utility and necessity of mathematical methods in this context.

This paper addresses the identification problem for a network composed of FHN models \citep{fitzhugh, nagumo}. The FHN model is traditionally regarded as a neuron model. However, it reflects the fundamental properties of excitability in general, making it applicable not only to modeling nerve cells but also to modeling neuron populations, brain regions \citep{gerster, chernihovskyi2005}, cardiac tissue \citep{takembo}, plant circadian rhythms \citep{Hartzell}, etc. It is important to note that in the majority of such applications networks of FHN models are of particular interest. This is especially true in the field of neuroscience, because the nervous system is a complex network of interconnected neurons, and processes occurring within it are directly related to its network structure. Thus, networks made up of single neuron models are of great interest. They can be used not only as models of some nerve cells populations but also as the whole brain model. 

Various methods for defining connections between nodes in FHN model networks have been explored in the literature. For instance, some studies consider networks with direct coupling as it was done by \cite{plotnikov_withdirect}, while others model neuron interactions using nonlinearities \citep{berner}, introduce delays in the equations \citep{masoliver, murguia} or use boundary coupling \citep{skrzypek}. Perhaps the most common type of coupling is diffusive coupling \citep{gerster, murguia, steur, omelchenko, hussain}, which is defined based on differences between corresponding state variables. In this work, we also employ diffusive coupling, adopting the general form proposed by \cite{omelchenko}. For example, \cite{gerster} used a diffusively coupled FHN network with 90 nodes to model the dynamics of the whole human brain during an epileptic seizure.

The problem addressed in this study assumes that only one state variable — specifically, the membrane potential or activator variable — of each FHN model is measurable. Furthermore, we account for potential inaccuracies in these measurements. These assumptions are natural; for instance, they reflect scenarios where EEG or ECoG signals are used as activator variables, as it was done by, for example, \cite{gerster, chernihovskyi2005, sevasteeva}.

Previous works have addressed similar problems. Very strong results in the sense of their applicability to a whole class of neural models were obtained by \cite{tyukin} and \cite{schiff}. The former solution is based on the application of observers, the Kalman filter behind the latter. However, the main advantage of these results may also be its weakness from some point of view: when working with systems whose structure is known, taking this structure into account during the identification algorithm design may lead to stronger results. Therefore, there is still an interest in new and possibly simpler solutions in this area. 

Numerous studies have specifically focused on identifying single FHN models. Most of these employ stochastic methods \citep{doruk, jensen}, particularly solutions based on the least squares method \citep{che, lou, wigren}. Other approaches, such as Kalman filter-based methods, have been proposed by \cite{azzalini, pagani, deng}. \cite{he} proposed very interesting identification method in which the measured variable is considered as the solution on the domain boundary of a function of several variables. Additionally, artificial neural networks have been used for parameter estimation in a few studies \citep{dong, rudi}. 

Our approach to the identification of FHN model networks relies on the application of a filter-differentiator, which addresses the practical unmeasurability of state variable derivatives, and the speed-gradient method \citep{fradkov_eng}. As a result, we derive a system of ordinary differential equations whose solutions converge to the true, albeit initially unknown, FHN parameter values corresponding to an input signal. In comparison with the majority of mentioned solutions, our results are supported by a theorem that provides the necessary and, as it seems to us, not burdensome conditions for asymptotic estimation achievement. Also, advantages of our method are its independence on the number of modeled neurons and is its applicability not only in the case of constant parameter values but also when parameters depend on time.

In terms of neural network model identification, \cite{zhou} considered this problem; however, identification in this context refers to determining the network topology, with the FHN model parameters assumed to be known a priori. Nevertheless, the problem of estimating FHN network parameters has not been previously solved or, even more so, mathematically investigated. 

\section{FHN neural network model}

Let us consider the dynamical system composed of $\mathrm N$ single FHN neuron models \citep{fitzhugh, nagumo}:
\begin{equation}
  \begin{cases}
    \dot u_k = u_k - \frac{u_k^3}{3} - v_k + I_{ext} + U_k,\\
    \dot v_k = \epsilon(u_k - a - bv_k) + V_k,
\end{cases} \label{FHN}
\end{equation}
where $k \in 1:\mathrm N$, $u_k(t)$ is the membrane potential of the $k$th neuron, $v_k(t)$ is the cumulative effect of all slow ion currents responsible for restoring the resting potential of the $k$th nerve cell membrane. We operate under the assumption that only membrane potential variables are measured, while $v_k$ and derivatives of all variables, $\dot u_k,\,\dot v_k$, cannot be measured. In addition, some measurement errors may be present that will be taken into account later.

Parameters $a$ and $b > 0$ determine the conductance characteristics of ion channels, $\epsilon > 0$ is the relative velocity change in slow ion currents. The main goal of this paper is to propose a way to determine values of these parameters using only measured variables, i.e. to identify the model \eqref{FHN}.

Parameter $I_{ext}$ is the known external current acting on neurons. 

Couplings between nodes are defined as differences between corresponding state variables with some constant coefficients: 
\begin{equation}
\begin{split}
    U_k = \sigma\sum\limits_{j = 1}^{\mathrm N}A_{kj}[B_{uu}&(u_j - u_k) \\&+ B_{uv}(v_j - v_k)],\\
    V_k = \sigma\sum\limits_{j = 1}^{\mathrm N}A_{kj}[B_{vu}&(u_j - u_k) \\&+ B_{vv}(v_j - v_k)],
\end{split}
\end{equation}
where $\sigma > 0$ is the overall coupling strength, $A_{kj}$ are elements of the network's adjacency matrix (we assume that the network graph is simple and undirected), $B_{uu},\,B_{uv},\,B_{vu}$ and $B_{vv}$ determine the interaction scheme between $u_k$ and $v_k$. Such type of coupling between network nodes is called diffusive and its properties have been widely studied \citep{steur, murguia}, but only in special cases when some or all of $B_{uv}, \,B_{vu}$ and $B_{vv}$ are equals to zero. Following papers of \cite{omelchenko} and \cite{gerster} we consider the general case which allows one to take into account not only direct couplings $u-u$ and $v-v$, but also cross-couplings between $u$ and $v$. 

Now let us transform the system \eqref{FHN} to more suitable form without unmeasurable variables and derivatives. First of all, to take systematic (scaling) errors into account it is assumed that $y_k=cu_k$, where $c > 0$ is an a priori unknown scaling factor enabling to take into
account systematic error occurring during the membrane potentials measurement and $y_k$, $k \in 1:\mathrm N$ are new variables. Then, each FHN model is transformed to the second-order differential equation:
\begin{equation}
   \ddot y_k = \theta_1^*\dot y_k + 3\theta_2^*y_k^2\dot y_k + \theta_3^*y_k +\theta_4^*y_k^3 +\frac1{\mathrm{N}}\theta_5^*,\label{FHN_mod1}
\end{equation}
where 
\begin{equation}
    \bm{\theta}^* = \begin{pmatrix}1 - \epsilon b \\ - c^{-2}/3 \\ \epsilon(b-1) \\ -\epsilon b c^{-2}/3 \\ \mathrm N c\epsilon(a+bI_{ext})\end{pmatrix}
\end{equation} is a vector of true parameter values that will be estimated later. Next, all models \eqref{FHN_mod1} are summed and the resulting equation is filtered to remove unmeasured derivatives. The second order filter-differentiator \citep{qwakernaak} is used:
$$
    \ddot x \approx p^2W(P)x,
$$
where the multiplier of $x$ is just transfer functions of the filter, $p = d/dt$, $W(p) = 1/(\tau_1p + 1)(\tau_2p + 1)$, $\tau_i>0$ $i = 1,2$; and new variables are obtained:
\begin{equation}
\begin{split}
y^* &= p^2W(p)\sum^{\mathrm N}_{k = 1} y_k,\\
x_1 &= pW(p)\sum^{\mathrm N}_{k = 1} y_k,\,
x_2 = pW(p)\sum^{\mathrm N}_{k = 1} y_k^3,\\
x_3 &= W(p)\sum^{\mathrm N}_{k = 1} y_k, \,
x_4 = W(p)\sum^{\mathrm N}_{k = 1} y_k^3.
\end{split} \label{new_variables}
\end{equation}

Therefore, the system \eqref{FHN} appears as follows:
\begin{equation}
    y^* = \theta_1^*x_1 + \theta_2^*x_2 + \theta_3^*x_3 + \theta_4^*x_4 + \theta_5^*.\label{FHN_mod}
\end{equation}

\begin{remark} \label{r1}
New variables in \eqref{FHN_mod} are outputs of stable linear systems \eqref{new_variables}. It means that boundedness of $y^*,\,x_1,\,x_2,\,x_3$ and $x_4$ follows from boundedness of inputs of these systems which are sums of $y_k$ and $y_k^3$. 
\end{remark}

\begin{remark} \label{r2}
Original parameters of the FHN model \eqref{FHN} is related with parameters $\theta^*$ as it follows:
\begin{equation}
\begin{split}
    a &= \frac{\theta^*_5\sqrt{-3\theta^*_2} - \mathrm N I_{ext}(1-\theta^*_1)}{\mathrm{N}(1 - \theta_1^* - \theta_3^*)},\\ 
    b &= \frac{1 - \theta_1^*}{1 - \theta_1^* - \theta_3^*},\;
    c = \frac1{\sqrt{-3\theta_2^*}},\\
    \varepsilon &= 1 - \theta_1^* - \theta_3^*. 
\end{split}\label{old_parameters}
\end{equation} 
\end{remark}

To determine the values of $\theta^*$ it is proposed in this paper to use the speed-gradient method based on the integral objective function which is described in detail and mathematically founded in Appendix \ref{secA1}. 

\section{Identification of FHN neural network model}

Let us transform the considered system \eqref{FHN}, \eqref{new_variables} to the general form \eqref{syst}:
\scriptsize
\begin{equation}
\begin{cases}
   \dot y_1 = y_1 + \theta_2^*y_1^3 -\frac1{\sqrt{-3\theta_2^*}}v_1 +\frac{I_{ext}}{\sqrt{-3\theta_2^*}} + Y_1',\\
   \begin{split}
   \dot v_1 = &\;(1-\theta_1^*-\theta_3^*)\sqrt{-3\theta_2^*}y_1 - \frac{\theta_5^*\sqrt{-3\theta_2^*}}{\mathrm N} \\&+ (\theta_1^* - 1)v_1 + V_1',
   \end{split}\\
   \ldots\\
   \begin{split}
   \dot y_{\mathrm N} = &\;y_{\mathrm N} + \theta_2^*y_{\mathrm N}^3 -\frac1{\sqrt{-3\theta_2^*}}v_{\mathrm N} + \frac{I_{ext}}{\sqrt{-3\theta_2^*}} \\&+ Y_{\mathrm N}',
   \end{split}\\
   \begin{split}
   \dot v_{\mathrm N} = &\;(1-\theta_1^*-\theta_3^*)\sqrt{-3\theta_2^*}y_{\mathrm N} - \frac{\theta_5^*\sqrt{-3\theta_2^*}}{\mathrm N} \\&+ (\theta_1^* - 1)v_{\mathrm N} + V_{\mathrm N}',
   \end{split}\\
   \dot w_1 = -\frac{\tau_1 + \tau_2}{\tau_1\tau_2}w_1 + w_2 - \frac{\tau_1 + \tau_2}{\tau_1^2\tau_2^2}\sum\limits^{\mathrm N}_{k = 1} y_k,\\
   \dot w_2 = -\frac1{\tau_1\tau_2}w_1 - \frac1{\tau_1^2\tau_2^2}\sum\limits^{\mathrm N}_{k = 1} y_k,\\
   \dot x_1 = w_1 + \frac1{\tau_1\tau_2}\sum\limits^{\mathrm N}_{k = 1} y_k,\\
   \dot x_2 = -\frac{\tau_1 + \tau_2}{\tau_1\tau_2}x_2 - \frac1{\tau_1\tau_2}x_4 + \frac1{\tau_1\tau_2}\sum\limits^{\mathrm N}_{k = 1} y_k^3,\\
   \dot x_3 = x_1,\\
   \dot x_4 = x_2,
\end{cases} \label{F}
\end{equation}
\normalsize
where \\$Y_k' = \sigma\sum_{j = 1}^{\mathrm N}A_{kj}\biggl[B_{uu}(y_j - y_k) + \frac{B_{uv}}{\sqrt{-3\theta_2^*}}(v_j - v_k)\biggl]$, $V_k' = \sigma\sum_{j = 1}^{\mathrm N}A_{kj}[\sqrt{-3\theta_2^*}B_{vu}(y_j - y_k) + B_{vv}(v_j - v_k)]$, $k \in 1:\mathrm N$. Here first $2\mathrm N$ rows model membrane potentials and depend on unknown parameters $\theta^*$. We mean that electrical signals of biological neurons (for example, EEG signals) will be used instead of this equations solutions, $y_k$, when setting up the network of FHN models in practice. The other six rows of the system \eqref{F} are equations of filters \eqref{new_variables} rewritten in differential form.

The equation \eqref{perfect_regression} corresponds to \eqref{FHN_mod}, i.e. $\bm z(\bm x) = (x_1 \quad x_2 \quad x_3 \quad x_4 \quad 1)^{\mathrm{T}}$. Thus, the adaptive law \eqref{theta(t)} can be used to identificate the network of FHN models. Let us duplicate it here and note that the proposed solution is a fifth-order system of ordinary differential equations (plus the six-order system for filtration):  
\begin{equation}
    \bm{\dot \theta} = -\bm\Gamma\delta(\bm x, \bm \theta)\bm z(\bm x),\label{theta(t)}
\end{equation}
where $\bm \theta(t)$ is a vector of estimates of $\bm \theta^*$, $\bm \Gamma \in \mathbb{R}^{5, 5}$ is a symmetric positive-definite gain matrix (for example, $\bm \Gamma = \mathrm{diag}\{\gamma_1,...,\gamma_5\}, \gamma_i > 0\;i\in1:5$),  $\delta(\bm x(t), \bm\theta(t)) = (\bm\theta(t)^{\mathrm{T}}\bm z(t) - y^*(t))$. These differential equations should be solved numerically and obtained solution should be used for the model \eqref{FHN_mod} or for the original model \eqref{FHN} with preliminary computation of original parameters using \eqref{old_parameters}. Now it remains to discuss the question of these estimates convergence to their true values. 

Let us denote by $r$ a number for which the following inequality is true:
\begin{equation}
\begin{split}
    R(\bm y, \bm v) = \frac1\sigma\sum_{k = 1}^{\mathrm N}(y_kY_k& + v_kV_k) \\&\leqslant r\sum\limits_{k = 1}^{\mathrm N}(y_k^2 + v_k^2).\label{0}
\end{split}
\end{equation}
There is a quadratic form in the vector $(\bm y\quad \bm v)^{\mathrm T}$ the left-hand side of this inequality, thus such $r$ always exists. The following theorem holds for the proposed solution:

\begin{theorem} \label{th}
    If the vector function of the FHN neural network model's observed values \eqref{FHN} $\bm z(\bm x) = (x_1 \quad x_2 \quad x_3 \quad x_4 \quad 1)^{\mathrm{T}}$ satisfies the persistent excitation (PE) condition \eqref{PE} and $\sigma < \varepsilon b/r$, then the identification law \eqref{theta(t)} guarantees that $\bm \theta(t) - \bm \theta^* \rightarrow 0 \,\, \mbox{for}\,\, t \rightarrow \infty$.
\end{theorem}

Proof of this theorem can be found in Appendix \ref{A2}.

\begin{remark} \label{r3}
The PE condition \eqref{PE} essentially means that the vector function $\bm z(t)$ do not converge to any hyperplane in $\mathbb{R}^m$ for $t \to \infty$. In other words, $\bm z(t)$ is rather diverse to assure the convergence of the parameters in \eqref{regression} to their true values. The complex dynamics of neural
processes gives hope that there would not be a problem with the presence of such diversity. The PE condition appear in many problems related to system identification and has been extensively studied by, for example, \cite{narendra}, \cite{fradkov_eng} and \cite{anderson}. The question of its numerical check is discussed in the following section.
\end{remark}

\begin{remark}\label{r4}
    For example, the maximum eigenvalue of the matrix corresponding to quadratic form $R(\bm y, \bm v)$ can be taken as r in \eqref{0}. It is easy to confirm that this matrix's eigenvalues are the set consist of connectivity graph Laplacian matrix's eigenvalues multiplied by $B_{uu}$ and Laplacian matrix's eigenvalues multiplied by $B_{vv}$.
\end{remark}

\section{Simulation results}

Now let us study the designed algorithm by computer simulation. For this purpose we take signals $y_k$ obtained with modeling of the network of FHN models with $\mathrm N = 5$ nodes and some assigned parameter values and initial data. Couplings between neurons are chosen as follows: 

\begin{figure}[!h]
\center{\includegraphics[width=1\linewidth]{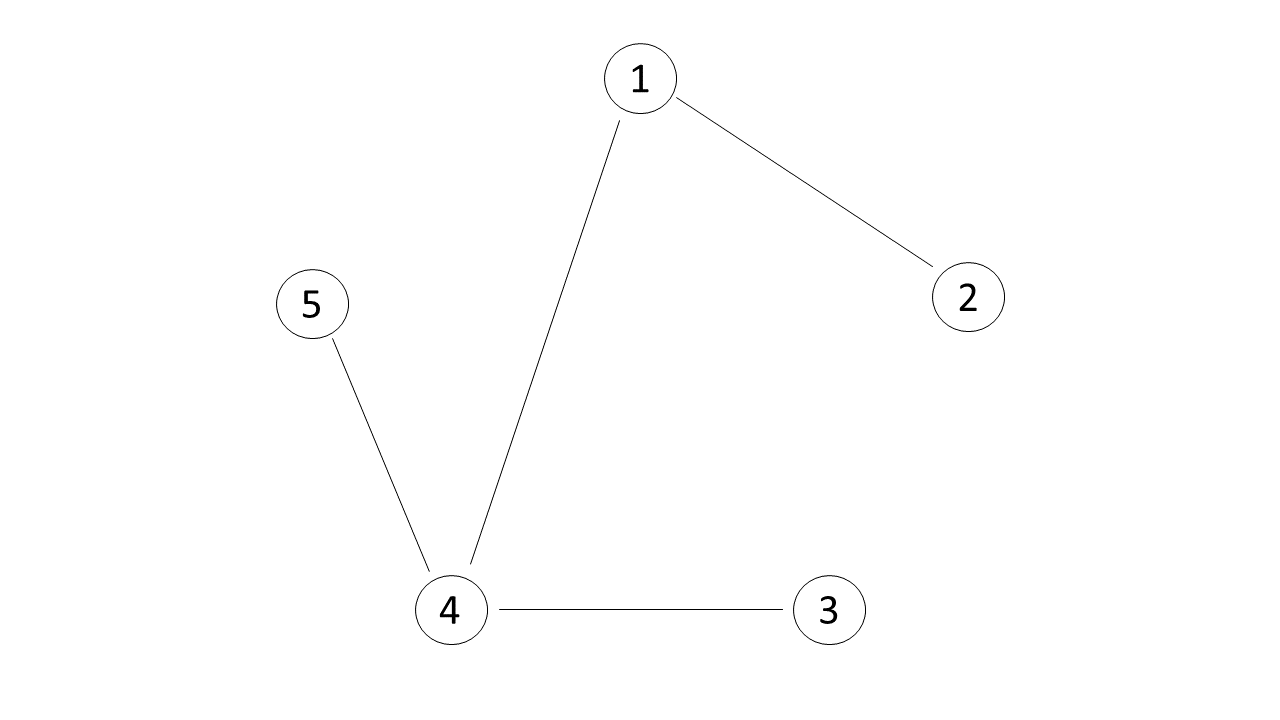} }
\caption{The topology of the considered network of five FHN models.}
\label{network_structure}
\end{figure}
\begin{equation}
\begin{split}
    &\sigma = 0.05, \;\phi = \frac{\pi}2 - 0.1,\\
    &B_{uu} = B_{vv} = \cos\phi, \\&B_{uv} = \sin\phi,\; B_{vu} = -\sin\phi.
\end{split} \label{network_params}
\end{equation}
Here we use the same way of couplings definition in FHN neural networks, which was suggested and investigated by \cite{omelchenko} and \cite{gerster}. The parameter values are chosen as follows:
\begin{equation}
\begin{split}
    &a = -0.7, \; b = 0.8, \; \varepsilon = 0.08, \; c = 1, \\ &I_{ext} = 1, \; \tau_1 = \tau_2 = 0.01, \; \bm \Gamma = \mathrm{\bm I};\label{params}
\end{split}
\end{equation}
and as initial data we take the following set of values:
\begin{equation}
\begin{split}
    \bm y(0) = (&0.7 \quad 0.1 \quad 0.9\quad -0.3\quad -0.6)^{\mathrm{T}},\\
    \bm v(0) = (&0.4 \quad 0.75\quad -0.1\quad -0.5\quad 0)^{\mathrm{T}},\\
    &\bm\theta(0) = \begin{pmatrix}0.98 \\ -0.275 \\ 0.005\\ -0.004\\ 0.066\end{pmatrix}.
\end{split} \label{initial_data}
\end{equation}

All of numerical experiments were made in the Simulink environment with the fourth-order Runge–Kutta method with step size $10^{-4}$.

First of all, let us numerically check fulfilment of Theorem \ref{th} condition which guaranties correct identification. This check reduces to the search of the integration interval length, L, at which the symmetric matrix $\bm M_L$ from \eqref{PE} is positive definite. Fig. \ref{pe_check} demonstrates that this condition is certainly met for $L \geqslant 4$. Let us note that with the growth of $L$ the smallest eigenvalue of $\bm M_L$ will increase too because of $\bm M_L$ structure. Therefore, $\bm z(t)$ is persistently excited for $L \geqslant 4$, small enough $\alpha$ and $t_0 > 0$ in the considered case.
 
\begin{figure}[!h]
\center{\includegraphics[width=1\linewidth]{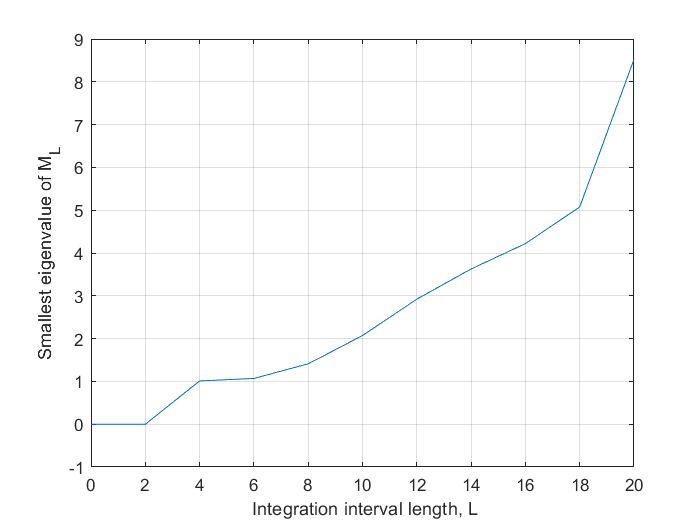} }
\caption{Dependence between the integration interval
length, $L$, and the smallest eigenvalue of the considered matrix, $\bm M_L$.}
\label{pe_check}
\end{figure}

In terms of the condition of $\sigma$ smallness, we can estimate the quadratic form $R(\bm y, \bm v)$ with maximum eigenvalue of its matrix, which equals to $0.42$ for the considered network, according to Remark \ref{r4}. Thus, if $\sigma < \varepsilon b/0.42 \approx 0.15$ then the Theorem \ref{th} conditions are fulfilled. 

The identification results are illustrated in Fig. \ref{experiment1} and in Table \ref{table1}. These results are rather consistent with our expectations: the identification goal is achieved. Furthermore, Table \ref{table1} shows that the error of $a, b, c, \varepsilon$ identification was reduced in more than 200 times (from approximately 0.81541 to 0.00358) by the algorithm \eqref{theta(t)}. If necessary the accuracy can be increased even more by the choice of a numerical method and its step size or by taking larger simulation time. 

\begin{figure}[!h]
\center{\includegraphics[width=1\linewidth]{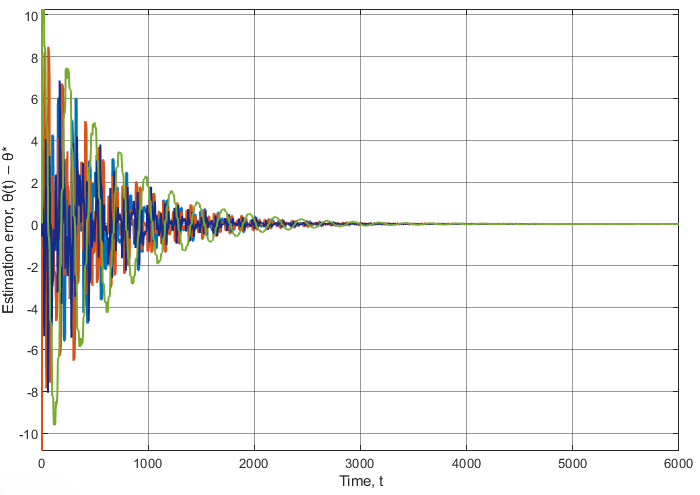} }
\caption{Parameter estimation errors, $\theta_i - \theta_i^*$, of the network of the FHN models \eqref{F} with couplings defined as on Fig. \ref{network_structure} and \eqref{network_params}, parameter values \eqref{params} and the set of initial data \eqref{initial_data}, obtained with the algorithm \eqref{theta(t)}.}
\label{experiment1}
\end{figure}

\begin{table}[h]
\caption{Values of original FHN parameters, $a,\,b,\,c,\,\epsilon$, corresponding (according to \eqref{old_parameters}) to initial data \eqref{initial_data}, $\bm \theta(0)$ and obtained estimations, $\bm \theta(6000)$, and theirs true values \eqref{params}.}\label{table1}%
\begin{tabular}{@{}llll@{}}
\toprule
 &$\bm \theta(0)$&$\bm \theta(6000)$&$\bm \theta^*$\\
\midrule
$a$ & -0.3 & -0.703224 & -0.7\\
$b$ & 1.5 & 0.801543 & 0.8\\
$c$ & 1.1 & 1.000189 & 1\\
$\epsilon$ & 0.01 & 0.079875 & 0.08\\
\botrule
\end{tabular}
\end{table}

Also, the choice of the gain matrix $\bm \Gamma$ in \eqref{theta(t)} can improve the identification accuracy and time. It may seem that we should take gains as big as it possible to increase the convergence rate. However, it is true in only perfect conditions of total lack of disturbances, but here we are inevitably dealing with a filtering error. This error eventually attenuates due to the stability of linear filters, but nevertheless  it has a significant impact on the system at the beginning of its operation. This is why not small enough gains cause occurrence of high amplitude oscillations of $\bm \theta(t)$. Fig. \ref{g_big} illustrates this effect.

\begin{figure}[!h]
\center{\includegraphics[width=1\linewidth]{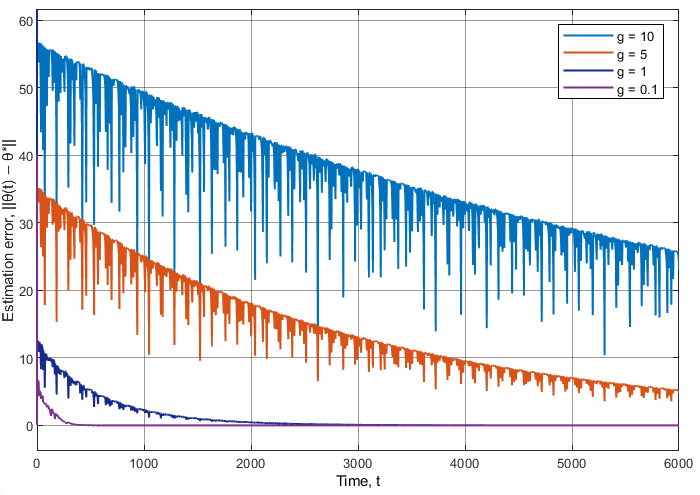} }
\caption{The norm of estimation error, $\|\bm \theta - \bm\theta^*\|$, of the network of the FHN models \eqref{F} with couplings defined as on Fig. \ref{network_structure} and \eqref{network_params}, parameter values \eqref{params} and the set of initial data \eqref{initial_data}, obtained with the algorithm \eqref{theta(t)}, in case of $\bm \Gamma = g\mathrm{\bm I}$, $g = 0.1, 1, 5, 10$.}
\label{g_big}
\end{figure}

On the other hand, too small gains cause identification slowdown as it can be seen at Fig. \ref{g_small}. Therefore, the choice of $\bm \Gamma$ is not a trivial task especially when working with real data in conditions of the lack of information. This choice can potentially be simplified by the speed-gradient algorithm robustification (or reducing its sensitivity to disturbances) as was done, for example, in \cite{cnn}.

\begin{figure}[!h]
\center{\includegraphics[width=0.8\linewidth]{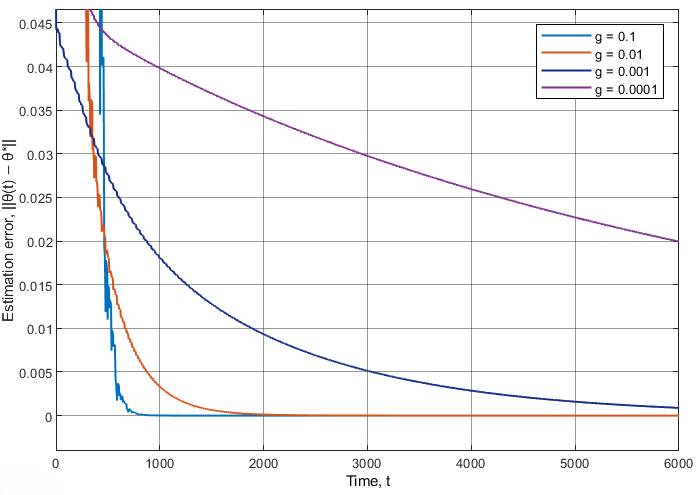} }
\caption{The norm of estimation error, $\|\theta_i - \theta_i^*\|$, of the network of the FHN models \eqref{F} with couplings defined as on Fig. \ref{network_structure} and \eqref{network_params}, parameter values \eqref{params} and the set of initial data \eqref{initial_data}, obtained with the algorithm \eqref{theta(t)}, in case of $\Gamma = g\mathrm{I}$, $g = 0.0001, 0.001, 0.01, 0.1$.}
\label{g_small}
\end{figure}

Now let us repeat the experiment from Fig. \ref{experiment1} with simpler couplings definition: 
\begin{equation}
    \sigma = 0.05, \; B_{uu} = 1, \;B_{vv} = B_{uv} = B_{vu} = 0; \label{network_params2}
\end{equation}
other parameter values:
\begin{equation}
\begin{split}
    &a = -0.525, \; b = 0.6, \; \varepsilon = 0.06, \; c = 0.75, \\ &I_{ext} = 1, \; \tau_1 = \tau_2 = 0.01, \; \bm\Gamma = \mathrm{\bm I};\label{params2}
\end{split}
\end{equation}
and other initial data for \eqref{theta(t)}:
\begin{equation}
\begin{split}
    y(0) = (&0.7 \quad 0.1 \quad 0.9\quad -0.3\quad -0.6)^{\mathrm{T}},\\
    v(0) = (&0.4 \quad 0.75\quad -0.1\quad -0.5\quad 0)^{\mathrm{T}},\\
    &\bm\theta(0) = \begin{pmatrix}0.98 \\ -0.353 \\ -0.08\\ -0.007\\ -0.339\end{pmatrix}.\
\end{split} \label{initial_data2}
\end{equation}
Results of this experiment are shown in Fig. \ref{experiment2} and Table \ref{table2}. Now the error of original parameters estimation has reduced in more than 7000 times (from approximately 0.59213 to 0.00008) during the time t = 6000. 

\begin{figure}[!h]
\center{\includegraphics[width=1\linewidth]{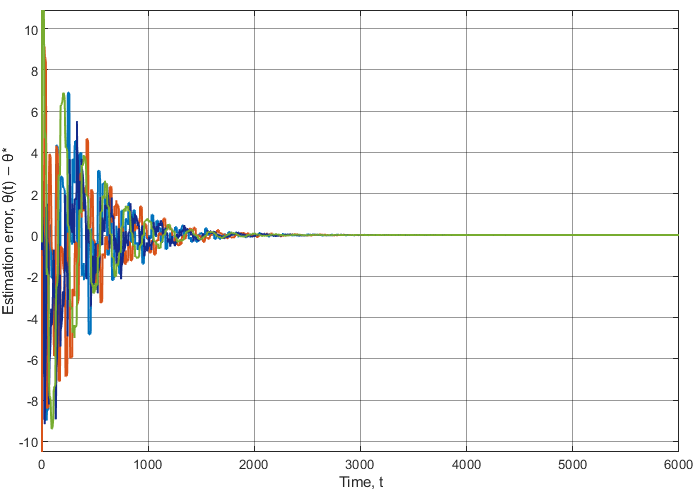} }
\caption{Parameter estimation errors, $\theta_i - \theta_i^*$, of the network of the FHN models \eqref{F} with couplings defined as on Fig. \ref{network_structure} and \eqref{network_params2}, parameter values \eqref{params2} and the set of initial data \eqref{initial_data2}, obtained with the algorithm \eqref{theta(t)}.}
\label{experiment2}
\end{figure}

\begin{table}[h]
\caption{Values of original FHN parameters, $a,\,b,\,c,\,\epsilon$, corresponding (according to \eqref{old_parameters}) to initial data \eqref{initial_data2}, $\bm\theta(0)$ and obtained estimations, $\bm\theta(6000)$, and theirs true values \eqref{params2}.}\label{table2}%
\begin{tabular}{@{}llll@{}}
\toprule
 &$\bm \theta(0)$&$\bm \theta(6000)$&$\bm \theta^*$\\
\midrule
$a$ & -0.9 & -0.525049 & -0.525\\
$b$ & 0.2 & 0.600064 & 0.6\\
$c$ & 0.97 & 0.750003 & 0.75\\
$\epsilon$ & 0.1 & 0.059992 & 0.06\\
\botrule
\end{tabular}
\end{table}

\section{\label{sec:level1}Conclusion}

In the paper, an algorithm for the parameter estimation of the FHN neural network, based on the speed-gradient method and the application of a filter-differentiator, was proposed, theoretically justified, and tested in numerical experiments. Its main advantage lies in accounting for the complex network structures that arise during the mathematical modeling of neural processes. Additionally, possible signal errors and the practical unmeasurebility of some variables and derivatives were taken into consideration. Despite this complexity, the proposed method remains relatively simple, as it only requires numerical solving of a system of several ordinary differential equations regardless of the number of modeled neurons. Also adaptive properties of this identification algorithm let to use it even when model parameters drift in time.

General results regarding the identification properties of the speed-gradient method, based on the integral objective function, were obtained in the paper. It was established that a sufficient condition for the FHN neural network identification using the speed-gradient algorithm is the persistent excitation (PE) of the observed signals. Numerical experiments confirmed that this condition was consistently met, suggesting that the complex dynamics of neural processes may contribute to the validity of PE in other scenarios as well. 

The simulation results illustrate the proposed solution. It was observed that the identification time of the FHN network is relatively short and that the accuracy of the parameter estimation is high. For instance, in one experiment, the parameter estimation error was reduced by more than 7000 times. Moreover, the accuracy and convergence rate can be improved by an appropriate choice of the gain matrix in the algorithm. However, this choice must be made carefully, as the gain matrix affects not only the convergence rate but also the amplitude of disturbances arising from the use of filters. In future work, it is planned to investigate robust versions of the algorithm to mitigate the influence of disturbances on the identification process.

The proposed approach can be applied to modeling the complex dynamics of nervous systems, particularly the human brain. For example, EEG signals can serve as inputs to the identification system, linking mathematical models of neural networks with real data collected by neurophysiologists. However, the obtained results are not strictly dependent on the specific model interpretation and, therefore, have a wide range of applications.

\backmatter

\bmhead{Acknowledgements}

This work was performed in IPME RAS and supported by  Russian Science Foundation, project 23-41-00060.

\section*{Declarations}

The authors declare no conflict of interest.

\begin{appendices}

\section{Identification using the speed-gradient method based on the integral objective function}\label{secA1}

Let us consider a general state space nonlinear system model:
\begin{equation}
    \bm{\dot x} = \bm F(\bm x, \bm \theta^*, t), \label{syst}
\end{equation}
where $\bm x(t) \in \mathbb{R}^n$ is a state vector, $\bm \theta^* \in \mathbb{R}^m$ is a vector of true (but unknown) parameter values, $t \geqslant 0$. 

Assume that $y^*(t) = y^*(x(t)) \in \mathbb{R}$, $\bm z(t) = \bm z(\bm x(t)) \in \mathbb{R}^m$ is the observable output of the system \eqref{syst}, for which the following equation holds:
\begin{equation}
    y^*(t) = \bm\theta^{*\mathrm{T}}\bm z(t). \label{perfect_regression}
\end{equation}
In other words, the system \eqref{syst} can be presented in a linear regression form \eqref{perfect_regression}. 

The problem of system \eqref{syst} identification is mathematically formulated in terms of \eqref{perfect_regression}: build an adaptive system with the estimate of the output variable of \eqref{syst}, $y(t) \in \mathbb{R}$, as the output variable and estimates of parameters of \eqref{syst}, $\bm \theta(t) \in \mathbb{R}^m$, as the parameters providing the identiﬁcation goal:
\begin{equation}
\begin{aligned} \label{goal}
1)&\,\,y(t) - y^*(t) \rightarrow 0 \,\, \mbox{for}\,\, t \rightarrow \infty,\\
2)&\,\,\bm\theta(t) - \bm\theta^* \rightarrow 0 \,\, \mbox{for}\,\, t \rightarrow \infty.
\end{aligned}
\end{equation}

To solve the problem, an adaptive system structure is chosen as follows:
\begin{equation}
    y(t) = \bm\theta(t)^{\mathrm{T}}\bm z(t).
    \label{regression}
\end{equation}
In order to tune the vector $\bm\theta(t)$ of parameter estimates the speed-gradient method based on the following integral objective function, $Q_t$, \citep{fradkov_eng} is used:
\begin{equation}
    Q_t = \int\limits_0^t \frac12 \delta^2(\bm x(s), \bm\theta(s)) ds,\label{Q}
\end{equation}
where $\delta(\bm x(t), \bm\theta(t)) = y(t) - y^*(t) = (\bm\theta(t) - \bm\theta^*)^{\mathrm{T}}\bm z(t)$. In this case the speed-gradient algorithm looks as follows: 
\begin{equation}
    \bm{\dot \theta} = -\bm\Gamma\delta(\bm x, \bm\theta)\bm z(x),\label{theta(t)_A1}
\end{equation}
where $\bm\Gamma \in \mathbb{R}^{m, m}$ is a symmetric positive-definite gain matrix.

\begin{theorem}\label{thA1}
Consider the system \eqref{syst} with the algorithm of unknown parameters estimation \eqref{theta(t)} obtained with the speed-gradient method based on the integral objective function \eqref{Q}.  Assume that the following conditions are met:\\
1. Vector functions $\bm F(\bm x, \bm\theta^*, t),\,\bm z(\bm x)$ are defined and continuous with their partial derivatives with respect to components of the vector $\bm x$.\\
2. $\bm F(\bm x(t), \bm \theta^*, t),\,\bm z(\bm x(t))$ are bounded for $0\leqslant t<\infty$. \\
3. $\bm z(t)$ satisfies the so called persistent excitation (PE) condition, i.e. there exist positive numbers $L, \alpha, t_0$ such that the following inequality holds for every $t > t_0$:
\begin{equation}
\bm M_L := \int\limits_t^{t+L}\bm z(s)\bm z(s)^{\mathrm{T}}ds \geqslant \alpha \bm I_m. \label{PE}
\end{equation}
Then, the identification goal \eqref{goal} in the system \eqref{syst}, \eqref{theta(t)} is achieved.
\end{theorem}

\begin{proof} 
The fulfillment of the first condition ensures the existence and uniqueness of the system \eqref{syst}, \eqref{theta(t)} solution for any initial conditions.

Let us define the Lyapunov function as $V_t = V(\bm \theta(t)) = Q_t + \frac12 \|\bm \theta(t) - \bm \theta^*\|^2_{\bm \Gamma^{-1}} = Q_t + \frac12 (\bm \theta(t) - \bm \theta^*)^{\mathrm T}\bm \Gamma^{-1}(\bm \theta(t) - \bm \theta^*)$. Its derivative with the respect to the system \eqref{syst}, \eqref{theta(t)}: $\dot V_t = \frac12 \delta^2(t) + (\bm \theta(t) - \bm \theta^*)^{\mathrm T}\bm \Gamma^{-1} \bm{\dot \theta}(t) = \frac12 \delta^2(t) - \delta(t) (\bm \theta(t) - \bm \theta^*)^{\mathrm T} \bm z(t) = -\frac12\delta^2(t) \leqslant 0.$  Therefore, firstly, the equilibrium point of the system \eqref{theta(t)} $\bm \theta^*$ is Lyapunov stable, secondly, $\frac12 \|\bm \theta(t) - \bm \theta^*\|^2_{\bm \Gamma^{-1}} \leqslant V_t \leqslant V_0 = \frac12 \|\bm \theta(0) - \bm \theta^*\|^2_{\bm \Gamma^{-1}}$ and $Q_t \leqslant V_0$, thus, $\bm \theta(t)$ is bounded and the finite limit $\lim\limits_{t\to \infty}\int\limits^t_0 \delta^2(s)ds$ exists. 

The expression $\frac{d}{dt}\delta^2(\bm x(t), \bm \theta(t)) = 2\delta(\bm x(t), \bm \theta(t))\dot\delta(\bm x(t), \bm \theta(t))=2(\bm \theta(t) - \bm \theta^*)\bm z(\bm x(t))(\bm {\dot \theta}(t)\bm z(\bm x(t))+(\bm \theta(t) - \bm \theta^*)\bm {\dot z}(\bm x(t)))$
is bounded because of the second condition and boundedness of $\bm \theta(t)$. Therefore, the function $\delta^2(t)$ is uniformly continuous. 

Thus, by the Barbalat lemma \citep{barbalat}, $\delta^2(t) \to 0$ for $t \to \infty$. It means that the first part of the identification goal \eqref{goal} is achieved: $y(t) - y^*(t) \rightarrow 0 \,\, \mbox{for}\,\, t \rightarrow \infty$.

Let us use a lemma from \citep{chaos} to prove the second identification goal \eqref{goal} achievement. 

\begin{proposition} \label{prA1}
Consider the smooth vector function $\bm w(t) \in R^{m}$ and the PE function $\bm z(t) \in R^{m}$, which is defined on $[0, \infty)$. If $\bm {\dot w}(t) \to 0$ and $\bm w(t)^T\bm z(t) \to 0$ for $t \to \infty$, then
\begin{equation}
    \lim_{t\to\infty}\bm w(t)=0.
\end{equation}
\end{proposition}

$\bm w(t)=(\bm \theta_1(t)-\bm \theta_1^* \quad \dots \quad \bm \theta_5(t)-\bm \theta_5^*)^{\mathrm{T}}$, $\bm z(t) = \bm z(\bm x(t))$ in our case. We have that $\bm w(t)^{\mathrm{T}}\bm z(t) = \delta(t) \to 0$ for $t\to\infty$; $\bm {\dot w}(t) = \bm {\dot\theta}(t) = - \bm \Gamma\delta(t)\bm z(t) \to 0$ for $t\to\infty$ because of the second condition. Thus, Proposition \ref{prA1} conditions are fulfilled and, therefore,  $\bm \theta(t) - \bm \theta^* \to 0$ for $t \to \infty$. 
\end{proof}

\section{Proof of Theorem \ref{th}}\label{A2}

Proof of Theorem \ref{th} is reduced to the check of the first two conditions of Theorem \ref{thA1} for the system \eqref{F}. The first condition is obviously met. Taking into account Remark \ref{r1} and the structure of \eqref{F}, fulfilment of the second condition is followed from boundedness of FHN trajectories.

\begin{proposition}
    If $\sigma < \varepsilon b/r$ then the solution of the system of $\mathrm N$ diffusively coupled FHN models is bounded. 
\end{proposition}

\begin{proof}
Let us define a non-negative function of $\bm y = (y_1 \quad y_2 \quad \ldots \quad y_{\mathrm N})^\mathrm{T}$ and $\bm v = (v_1 \quad v_2 \quad \ldots \quad v_{\mathrm N})^\mathrm{T}$ as follows: 
\begin{equation}
    H(\bm y, \bm v) = \frac12\sum_{k = 1}^{\mathrm N} (y_k^2 + v_k^2). \label{H}
\end{equation}
Its derivative with respect to the system \eqref{F} is as follows: 
\begin{equation}
\begin{split}
    &\dot H(\bm y, \bm v) = \sum_{k = 1}^{\mathrm N}\Bigl(-\alpha_2y_k^4 - \alpha_1 v_k^2 + y_k^2\\& + c_1y_kv_k +c_2y_k  + c_3v_k\Bigl) + \sigma R(\bm y, \bm v),
\end{split} \label{dot_H}
\end{equation}
where 
\begin{equation}
\begin{split}
    &\alpha_2 = -\theta_2^* = c^{-2}/3 > 0,\\
    &\alpha_1 = 1 - \theta_1^* = \varepsilon b > 0,\\
    &c_1 = - (3\theta_2^*(1 - \theta_1^* - \theta_3^*)+1)/\sqrt{-3\theta_2^*},\\
    &c_2 = I_{ext}/\sqrt{-3\theta^*_2},\\
    &c_3 = -\theta_5^*\sqrt{-3\theta_2^*}/\mathrm N.
\end{split}
\end{equation}

Last three terms in parenthesis in \eqref{dot_H} can be estimated by standard quadratic inequalities:
\begin{equation}
\begin{split}
    &c_1y_kv_k \leqslant \frac{c_1^2}{2\delta_1}y_k^2 + \frac{\delta_1}{2}v_k^2,\\
    &c_2y_k \leqslant \frac{\delta_1}{2}y_k^2 + \frac{c_2^2}{2\delta_1},\;c_3v_k \leqslant \frac{\delta_1}{2}v_k^2 + \frac{c_3^2}{2\delta_1},
\end{split} \label{1}
\end{equation}
where $\delta_1$ is an arbitrary positive constant. Similarly for any $\delta_2 > 0$ the following inequality holds:
\begin{equation}
    y_k^2 \leqslant \frac{\delta_2}{2}y_k^4 + \frac1{2\delta_2}. \label{y_k^2}
\end{equation}
Multiplying both sides of \eqref{y_k^2} by $2\alpha_2/\delta_2 > 0$ leads to the following estimate of first term in parenthesis in \eqref{dot_H}:
\begin{equation}
    -\alpha_2y_k^4 \leqslant -\frac{2\alpha_2}{\delta_2}y_k^2 + \frac{\alpha_2}{\delta_2^2}. \label{2}
\end{equation}

From \eqref{0}, \eqref{1}, \eqref{2} it implies that \eqref{dot_H} can be estimated as follows:
\begin{equation}
\begin{split}
    &\dot H(\bm y, \bm v) \leqslant -\beta_1^y\sum_{k = 1}^{\mathrm N}y_k^2 -\beta_1^v\sum_{k = 1}^{\mathrm N}v_k^2 + \beta_2,
\end{split} \label{dot_H_ineq}
\end{equation}
where 
\begin{equation}
\begin{split}
    &\beta_1^y = \frac{2\alpha_2}{\delta_2} - 1 - \frac{c_1^2}{2\delta_1} - \sigma r - \frac{\delta_1}2,\\
    &\beta_1^v = \alpha_1 - \sigma r - \delta_1,\\
    &\beta_2 = \mathrm N\left(\frac{c_2^2 + c_3^2}{2\delta_1} + \frac{\alpha_2}{\delta_2}\right).
\end{split}
\end{equation}
By the Lemma condition $\sigma < \varepsilon b/r = \alpha_1/r$, therefore there are such $\delta_1 > 0$, $\delta_2 > 0$ that $\beta_1^y$ and $\beta_1^v$ are positive. \eqref{dot_H_ineq} can be rewritten as follows:
\begin{equation}
    \dot H(\bm y, \bm v) \leqslant -\beta_1H(\bm y, \bm v) + \beta_2, \label{dot_H_ineq2}
\end{equation}
where $\beta_1 = \mathrm{min}\{\beta_1^y,\,\beta_1^v\} > 0$. Integrating of the inequality \eqref{dot_H_ineq2} leads to the following estimate:
\begin{equation}
    H(\bm y(t), \bm v(t)) \leqslant H(\bm y(0), \bm v(0))\mathrm{e}^{-\beta1} + \frac{\beta_2}{\beta_1}.
\end{equation}

Therefore, $H(\bm y, \bm v)$ is a bounded function, which means that the norm of the vector $(\bm y\quad \bm v)^{\mathrm T}$ is bounded.
\end{proof}




\end{appendices}


\bibliography{sn-bibliography}% common bib file

\end{document}